\documentclass[11pt]{article}
\usepackage{paralist}
\usepackage{color}
\usepackage{bm}

\usepackage{soul}

\usepackage[cm]{fullpage}
\usepackage{amsfonts}
\usepackage{amssymb}

\usepackage{amsmath,enumitem}
\usepackage[usenames,dvipsnames]{xcolor}
\usepackage{tikz,pgf,color}
\usepackage{constants}
\usepackage{amsthm}

\usepackage[breaklinks]{hyperref}

\usepackage[margin=1in]{geometry}

\theoremstyle{plain}
\newtheorem{theorem}{Theorem}[section]

\newtheorem{lemma}[theorem]{Lemma}

\newtheorem{definition}[theorem]{Definition}

\newcommand{\calU}{\ensuremath{\mathcal{U}}}
\newcommand{\calP}{\ensuremath{\mathcal{P}}}
\newcommand{\calD}{\ensuremath{\mathcal{D}}}
\newcommand{\cc}{\ensuremath{\mathrm{cc}}}
\newcommand{\BE}{\ensuremath{\mathrm{BE}}}
\newcommand{\RE}{\ensuremath{\mathrm{RE}}}
\newcommand{\dt}{\ensuremath{\mathrm{dt}}}
\newcommand{\CD}{\ensuremath{\mathrm{CD}}}

\newcommand{\D}{\ensuremath{\widetilde{D}}}

\newcommand{\Stream}{\ensuremath{\mathcal{S}}}
\usepackage{algorithm}
\usepackage{algpseudocode}
\usepackage[normalem]{ulem}

\algtext*{EndWhile}
\algtext*{EndIf}
\algtext*{EndFor}

\newcommand{\rom}[1]{%
	\textup{\uppercase\expandafter{\romannumeral#1}}%
}

\newcommand{\CT}{\mathrm{CT}}
\newcommand{\E}{\mathrm{E}}
\newcommand{\LT}{\mathrm{LT}}

\newcommand{\radius}{\mathrm{radius}}

\newcommand{\dep}{\mathrm{depth}}

\newcommand{\lab}{\mathrm{label}}
\newcommand{\Bi}{\mathrm{Bi}}

\newcommand{\calH}{\mathcal{H}}
\newcommand{\calT}{\mathcal{T}}

\newcommand{\calE}{\mathcal{E}}
\newcommand{\Good}{\textbf{Good}}
\newcommand{\Bad}{\textbf{Bad}}

\newcommand{\junk}[1]{{}}
\newcommand{\disc}{\ensuremath{\mathrm{disc}}}
\newcommand{\Var}{\textrm{Var}}

\newcommand{\poly}{\ensuremath{\mathrm{poly}}}

\allowdisplaybreaks

\begin{document}

\title{\Large Estimating Graph Parameters from Random Order Streams}
\author{Pan Peng\thanks{Faculty of Computer Science, University of Vienna, Austria. 
		Email: \url{pan.peng@univie.ac.at}. The research leading to these results has received funding from the
		European Research Council under the European Union's Seventh
		Framework Programme (FP/2007-2013) / ERC Grant Agreement no. 340506.} \\
\and
Christian Sohler\thanks{Department of Computer Science, TU Dortmund, Germany. 
	Email: \url{christian.sohler@tu-dortmund.de}. Supported by ERC Starting Grant 307696.}}
\date{}

	\begin{titlepage}
	
	\maketitle
	
	\thispagestyle{empty}






\begin{abstract} 
We develop a new algorithmic technique that allows to transfer some constant time 
approximation algorithms for general graphs into random order streaming algorithms.
We illustrate our technique by proving that in random order streams with probability 
at least $2/3$,
\begin{itemize}
	\item
	the number of connected components of $G$ can be approximated up to an additive error 
	of $\varepsilon n$ using $(\frac{1}{\varepsilon})^{O(1/\varepsilon^3)}$ space,
	\item
	the weight of a minimum spanning tree of a connected input graph with integer edges weights
	from $\{1,\dots,W\}$ can be approximated within a multiplicative factor of $1+\varepsilon$ using
	$\big(\frac{1}{\varepsilon}\big)^{\tilde O(W^3/\varepsilon^3)}$ space,
	\item
	the size of a maximum independent set in planar graphs can be approximated within a multiplicative
	factor of $1+\varepsilon$ using space $2^{(1/\varepsilon)^{(1/\varepsilon)^{\log^{O(1)} (1/\varepsilon)}}}$.
\end{itemize}
\end{abstract}
\end{titlepage}

\section{Introduction}

The analysis of very large networks is still a big challenge. In particular, if the size of the network is so big that it does not fit
into the main memory of a computer and the edges appear as a data stream, we would like to be able to approximately analyze the network
structure. If the network is relatively dense, i.e. the number of edges is significantly bigger than $n \log^{O(1)} n$, where $n$ is the
number of vertices of the network, then semi-streaming algorithms, i.e. algorithms that process the edges in worst-case (or adversarial) order and
use $n \log^{O(1)} n$ memory provide an important space efficient (relative to the input size) approach\footnote{Throughout the paper, we will focus on single-pass streaming algorithms.}~\cite{HRR99:stream,FKMSZ05:graph,FKM08:distance}. 
However, for sparse networks it seems to be reasonable to require that a space efficient algorithm uses space $o(n)$.
In this case, if we only aim at slightly sublinear space, one can still approximate some network parameters \cite{BS15:matching,HP16:stream}. However,
if we would like our space requirement to be, say, $\log^{O(1)} n$ then it seems that only very few network parameters can
be approximated in the classical worst-case order model. One reason is that pointer chasing is hard in graph streams, but sparse graphs are typically explored by following paths \cite{FKM08:distance}. 

One way to still be able to develop some algorithms with reasonable guarantees
is to parametrize the input: This has, for example, been used to characterize the space complexity of
algorithms for approximating the number of specific small subgraphs of the input graph \cite{BKS02:reductions}. Also, there are algorithms available 
that parametrize the input in terms of the size of the optimal structure that is to be computed \cite{CCHM14:parameterized}. One might also be able to obtain space efficient streaming algorithms for some special class of graphs. For example, there exist constant approximation algorithms that use $O(\log n)$ words of space and approximate the size of the maximum matching for graphs with bounded aboricity~\cite{MV16:note,CJMM16:matching}.   

Another way is to relax the assumption that the edges come in worst-case order. A natural relaxation is to assume that the edges come
in random order, that is, the input stream is chosen uniformly at random from the set of all possible permutations of the edges~(e.g., \cite{CCM08:lower,KMM12:matching,KKS14:matching}). In a recent survey~\cite{McG14:stream}, McGregor asked the question if one can take advantage of the random order assumption to obtain more space-efficient streaming algorithms. In general, the answer to this question is unclear and heavily depends on the problem under consideration. On the negative side, it is known that to distinguish if a graph is connected or not, any random order streaming algorithm needs $\Omega(n)$ bits of space~\cite{CCM08:lower}, which nearly matches the corresponding  $\Theta(n\log n)$ upper/lower bounds in adversary order model~\cite{FKMSZ05:graph,SW15:tight}. On the other hand, there exists work that in some sense addresses the above question in the affirmative. For example, it is possible to approximate the size of a maximum matching upto a polylogarithmic factor 
in polylogarithmic space \cite{KKS14:matching} in the random order model, while there is no such algorithm known in the adversarial order model. More recently, together with Monemizadeh and Muthukrishnan, the authors showed that for graphs with
\emph{maximum degree bounded by a known constant $d$} any constant time property testing algorithm can be simulated in random order
streams using constant space~\cite{MMPS17:constant}.  

In this paper we will also follow the approach from \cite{MMPS17:constant} of transforming sublinear time algorithms into streaming algorithms for random order streams. However, in contrast to \cite{MMPS17:constant}, our work will focus on \emph{general graphs} without degree bound as well as classes of graphs such as planar graphs with \emph{bounded average degree}. The main finding of this paper is a new algorithmic technique (and its analysis) that allows to transfer many results from the area of 
sublinear time algorithms for general graphs and classes of graphs with bounded average degree to random order graph streaming algorithms. We illustrate the use of this technique on the examples of approximating the number of connected components, the weight of the minimum spanning tree of an arbitrary graph and the size of an independent set in a planar graph. 


\subsection{Our Results}

Our first result is an algorithm for approximating the number of connected components of the graph within additive error $\varepsilon n$. Testing connectivity and approximating the number of components of a graph is a basic primitive in the area of sublinear time algorithms and has been studied significantly \cite{GR02:testing,BKM14:NCC} and used as a subroutine in other sublinear algorithms such as \cite{CRT05:MST,CS09:metricMST}. 

\begin{theorem}
	Let $0 < \varepsilon < 1/2$. Then there is an algorithm that takes as input a random order stream of edges from a graph $G=(V,E)$ and computes an output value $\hat c$ that with probability at least
	$2/3$ approximates the number $c$ of connected components of $G$ with an additive error of $\varepsilon n$,
	i.e. we have
	$$
	| \hat c -c | \le \varepsilon n.
	$$
	The algorithm uses $\big(\frac{1}{\varepsilon}\big)^{O(\frac{1}{\varepsilon^3})}$ words of space.
\end{theorem}

In comparison to the above result, Chakrabarti et al.~\cite{CCM08:lower} presented a $\Omega(n)$ space lower bound for any single-pass algorithm that tests if a graph is connected or not in random order streams, which suggests that approximating the number of connected components within a multiplicative factor $2-\delta$ will already require $\Omega(n)$ space, for any $\delta>0$. 

Following the ideas of \cite{CRT05:MST} an approximation algorithm for the number of connected components
can be transformed into an algorithm for approximating the weight of the minimum spanning tree of 
the input graph when we have integer weights from a bounded set $\{1,\dots, W\}$ (using a simple
rounding procedure this can also be extended at the cost of a slightly higher space complexity to real weights from the same range).

\begin{theorem}~\label{thm:weight_mst}
	Let $0 < \varepsilon < 1/2$. Then there exists an algorithm that takes as input a random order stream
	of weighted edges from a connected graph $G=(V,E)$ whose edge weights are from the universe $\{1,\dots,W\}$
	and with probability at least $2/3$ computes an estimate $\hat M$ 
	for the weight $M$ of the minimum spanning tree of $G$ that satisfies
	$$
	(1-\varepsilon) M \le \hat M \le (1+\varepsilon) M.
	$$
	The algorithm uses $\big(\frac{1}{\varepsilon}\big)^{\tilde O(W^3/\varepsilon^3)}$ words of space.
\end{theorem}

It is worth pointing out that in the worst-case order model, there are lower bounds $\Omega(n^{1-O(\varepsilon)})$ and $\Omega(n^{1-O(\varepsilon/W)})$ on the space complexity for the above two problems, respectively~\cite{HP16:stream}. In contrast, our random order streaming algorithms use space independent of $n$.

Our third result is a random order streaming algorithm for $(1+\varepsilon)$-approximating the size of maximum independent set in minor-free graphs (which include planar graphs).
\begin{theorem}~\label{thm:planar_mis}
	Let $G=(V,E)$ be a $K_t$-minor-free graph for any constant $t>0$. Given a parameter $0<\varepsilon<1/2$, there exits a random order graph stream algorithm that $(1+\varepsilon)$-approximates the size of the maximum independent set of $G$ with probability at least $2/3$ while using 
	$2^{(1/\varepsilon)^{(1/\varepsilon)^{\log^{O(1)} (1/\varepsilon)}}}$ words of space. 
\end{theorem}

We remark that the above result is obtained by applying our new algorithmic technique to approximate the frequency of non-isomorphic copies of bounded-degree neighborhoods (see also discussions in Section~\ref{sec:comparison}). Our approach here is very general and can be also extended to transform many existing constant-time approximation algorithms (that are given query access to the adjacency list of the graph) for graphs with \emph{bounded average degree} into constant-space random order streaming algorithms. 
%
%
%
%

Finally, we would like to mention that our newly developed technique can also be used to derive the main result of~\cite{MMPS17:constant}, i.e., any bounded degree graph property that is constant-query testable in the query access (to the adjacency list) model can be tested with constant space in random order streaming model. This is true since to test these properties in bounded degree graphs, it suffices to approximate the frequency of non-isomorphic copies of some neighborhoods~\cite{CPS16:testing,GR11:proximity}. 
We omit further details for this result here.

%


\subsection{Our Techniques}\label{sec:techniques}
A major obstacle in designing graph streaming algorithms is that standard graph exploration 
algorithms like Breadth First Search (BFS) and Depth First Search (DFS) suffer from the fact that the edges typically do not appear in the 
order of the graph traversal and thus may be missed. This makes it impossible for many problems to design graph streaming algorithms that in the \emph{worst-case} require only constant or polylogarithmic space. If the 
edges arrive in \emph{random order} there is a small probability that they appear in the same order 
as in a given graph traversal. However, typically this will not be the case as well. One of the main challenges for designing random order streaming algorithms can therefore be viewed as the task
to identify when the graph traversal behaves as in the original graph and when it does not.
If we can solve this problem we can often simulate constant time algorithms in random order streams. 

The algorithmic technique for designing streaming algorithms from constant time algorithms, which we develop in this paper, exploits that many sublinear time algorithms may be viewed as 
computing their output from an estimation of the frequency of certain subgraphs of interest like, for example, connected components or $k$-discs (i.e., the subgraphs induced by vertices at distance at most $k$ from any given vertex).
Since the algorithms require constant time, we also know that the subgraphs of interest are of constant 
size. For simplicity, we will from now on assume that we are looking for subgraphs that have 
exactly $k$ vertices. In the following we will use the example of approximating the number of connected
components of $G$ to illustrate our ideas. 
In this case, subgraphs of interest are connected graphs on exactly $k$ vertices that have no outgoing edges, i.e. connected components of size $k$.

Let us consider such a connected component $C$
with vertex set $V_C$ and edge set $E_C$ and let $v\in V_C$ denote a vertex of $C$. As mentioned above, in random order streams with 
some small probability (over the choice of the random order), the edges of $C$ arrive in the order such that a greedy graph traversal (whenever a new
edge connects our current connected component to a new vertex, we add this new vertex to the current connected component) from $v$ will see the 
whole component. The challenging part is that we cannot distinguish the case that the greedy traversal finds all vertices of a connected component 
from the case that $v$ belongs to a larger connected component and the greedy traversal only discovers edges incident to $k$ vertices (this may happen because 
other edges arrive in the stream before their endpoints have been discovered by the greedy traversal starting from $v$). The only way
to still be able to estimate the number of connected components is to use more involved statistics of the discovered subgraphs.
Our previous work \cite{MMPS17:constant} provided a solution for graphs whose maximum degree is bounded by a known constant $d$. However, this solution does 
not seem to be extendable to the setting where no degree bound is available (see more discussions in Section~\ref{sec:comparison}).

Our new algorithm is parametrized by a small constant $\lambda>0$ and divides the stream into two parts: The first $\lambda m$ edges (called the \emph{first phase}) and the remainder of the stream (called the \emph{second phase}), where $m$ denotes the overall number of edges. We then check, whether the structure we are looking for (for example, a fixed spanning tree $T$ of a connected component of size $k$) is observed by a greedy graph traversal from some vertex $v$ using only edges of the first phase. If so, it is possible that $T$ belongs to a connected component of size larger than $k$, i.e., $T$ (or $v$) is a false positive, in which case we still have the same probability to observe $T$. However, there is also at least one edge $e$ that witnesses that $T$ does not span a connected component of size $k$. In order to find this edge and reveal $T$ as a false positive, we scan the edges of the second phase. The interesting observation is that the conditional probability that $e$ is not in the second phase of the stream is 
$$
\Pr\{e \in M \; | \; T \subseteq M\} \approx \lambda,
$$
where $M$ denotes the set of edges in the first phase.  
Therefore, it is unlikely that we accidentally keep a false positive (if $\lambda$ is sufficiently small). 
Using a more complex form of the above observation (that also takes the order of edges in $T$ from the stream into account) we can prove that the conditional probability of discovering $T$ but not seeing $e$ lateron in the
stream is at most $\lambda$.
The main challenges for the analysis are that
\begin{enumerate}[label={(\alph*)}]
	\item\label{item:challenge_a}
	in the context of vertices with unbounded degree, the number of trees of size $k$ that are rooted at a given vertex can be a function of $n$. Since every such tree is potentially a false positive we need to 
	argue carefully that the probability for a false positive is sufficiently small.
	\item\label{item:challenge_b}
	a more subtle point is that the probability of finding a spanning tree of a connected component
	depends on the structure of the connected component. This implies in particular, that using the
	approach sketched above the distribution of the vertices that are accepted after the first phase
	is no longer uniform, which makes it difficult to obtain an estimate for the number of connected
	components from this. In order to address this problem, we will only consider a very specific
	type of spanning tree, which we call canonical BFS tree.
\end{enumerate}

Finally, let us observe that our technique is not specifically tailored to the problem of approximating
the number of connected components. We show that it can also be used to approximate the distribution of
$k$-discs in the graph $G'=(V,E')$ that is obtained from the input planar graph (or graph with bounded average degree) by removing all 
vertices with degrees that exceed a certain constant threshold $t$. It is known that a planar graph has $O(n/t)$
such vertices and a maximum independent set of size $\Omega(n)$.
From this distribution we know the structure of the graph upto $\varepsilon n$ edges \cite{NS13:hyperfinite,HKNO09:local} and we can therefore estimate the size of a maximum independent set.
Obviously, this approach immediately translates to any other problem whose objective value is
$\Omega(n)$ in planar graphs and does not change by more than $O(1)$ when a vertex is deleted or an 
edge is deleted or inserted. We believe that our approach can also be extended to a number of other 
problems.

\subsection{Comparison to \cite{MMPS17:constant}}\label{sec:comparison}

In our previous work with Monemizadeh and Muthukrishnan we provided a technique to transform
a constant time property tester for bounded degree graphs into a random order streaming algorithm
(the technique also works for constant-time algorithms). The central idea was to show that
in random order streams one can approximate the distribution of $k$-discs when the input graph has bounded maximum degree in the following way (which builds on earlier ideas from \cite{CPS16:testing}).
We sample a set of vertices uniformly at random and consider the $k$-discs observed by doing a BFS 
in a random order stream. Due to the random order we will typically miss edges, so any observed
$k$-disc is typically a subgraph of the true $k$-disc. However, since we have maximum degree bounded by $d$, there is a maximal $k$-disc, i.e. a $k$-disc in which every vertex has degree $d$ and we can
identify, if we have sampled all edges from such a disc. This can be used to approximate the frequency of maximal $k$-discs. Furthermore, for any $k$-disc type we can approximate the probability that we observe this type
given that the true $k$-disc is maximal. Using this together with the observed number of maximal
$k$-discs we can compute the number of $k$-discs that have one edge less than a maximal $k$-disc
from the observed number of such $k$-discs. By defining a suitable partial order on $k$-discs
one can use this approach to approximate the whole distribution of $k$-discs, which can be 
used to simulate property testers or constant time algorithms with query complexity at most $k$.

The above approach relies crucially on the assumption of a known degree bound and we do not believe
that it can be generalized to arbitrary graphs. This is indicated by the fact that in the property testing setting \cite{CPS16:testing}, where a similar approach had been used to approximate the distribution of
directed $k$-discs (in a scenario where only outgoing edges can be seen from a vertex), it has been shown that one cannot extend this approach to graphs with bounded average degree and obtain similar guarantees.

We also remark that in the full version of~\cite{MMPS17:constant}, we developed a 2-pass algorithm that approximates the distribution of 
$k$-discs in bounded degree graphs using an approach that on a high level is somewhat similar 
to our approach in the sense that it uses the first pass to sample a $k$-disc and the second pass to
verify, if edges have been missed. However, if one allows two passes and the graph is of bounded degree, then the verification step is easy
while for one pass and graphs with arbitrarily large maximum degree this is one of the main challenges.

\subsection{Other Related Work}
For random order streams, Chakrabarti et al.~\cite{CCM08:lower} gave $\Omega(n)$ space lower bound for any single-pass algorithm that tests if a graph is connected or not and $\Omega(n^{1+1/t})$ space lower bound for any single-pass algorithm that tests if two given vertices have distance at most $1$ or at least $t+1$. Konrad et al.~\cite{KMM12:matching} gave random order streaming algorithms with space complexity $\tilde{O}(n)$ for the maximum matching problem for bipartite graphs and general graphs with approximation ratio better than $2$. Kapralov et al.~\cite{KKS14:matching} gave a $\poly\log n$-approximation algorithm for the size of maximum matching of general graphs using $\poly\log n$ space in random order streams. For graphs with bounded arboricity, Esfandiari et al.~\cite{EHLMO14:matching} have given a constant approximation algorithm for the size of maximum matching using $\tilde{O}(\sqrt{n})$ space in random order streams. Later, an improved algorithm with space complexity $O(\log n)$ that even holds for adversary order streams has been given \cite{CJMM16:matching,MV16:note}. 

In the following, we use the term $(\alpha,\beta)$-approximation algorithm to denote an algorithm that approximates the
objective with a guarantee of $\alpha \textrm{Opt} + \beta$, where $\textrm{Opt}$ denotes the cost of an optimal solution.
Chazelle et al.~\cite{CRT05:MST} gave a $(1,\varepsilon n)$-approximation algorithm in the \emph{query access} (to the adjacency list) model for the number of connected components that runs in time $O(\bar{d}\varepsilon^{-2}\log\frac{\bar{d}}{\varepsilon})$, and a $(1+\varepsilon)$-approximation algorithm with running time $O(\bar{d}W\varepsilon^{-2}\log\frac{\bar{d}W}{\varepsilon})$ for the minimum spanning tree weight for graphs with average degree $\bar{d}$ and edge weights in the set $\{1,\cdots, W\}$. 
Berenbrink et al.~\cite{BKM14:NCC} showed that for graphs without multiple edges, one can $(1,\varepsilon n)$-approximate the number of connected components in time $O(\frac{1}{\varepsilon^2}\log (\frac{1}{\varepsilon}))$. 

There exists $(1,\varepsilon n)$-approximation algorithm in the query access model for estimating the sizes of a minimum vertex cover, maximum dominating set, and maximum independent set for any class of graphs with a fixed excluded minor with running time $2^{\poly(1/\varepsilon)}$~\cite{HKNO09:local}. Note that the maximum degree of such a graph might be arbitrarily large, while the average degree is always bounded by a constant. 

Onak et al.~\cite{ORRR12:vetexcover} gave a $(2,\varepsilon n)$-approximation for the size of minimum vertex cover with running time $\tilde{O}(\bar{d}\cdot \poly(1/\varepsilon))$, which improves upon previous results by Parnas and Ron~\cite{PR07:sublinear_distributed}, Marko and Ron~\cite{MR09:distance}, Nguyen and Onak~\cite{NO08:constant} and Yoshida, Yamamoto and Ito~\cite{YYI12:constant}. 

Other parameters of the sparse graphs can also be approximated with running time $f(\bar{d},\varepsilon)$, which is independent of the input size $n$ if $\bar{d}$ is bounded by some constant. (Most of such results are stated in terms of bounded maximum degree $d$, while they can generalized to graphs with bounded average degree $\bar{d}$ by ignoring all vertices with degree greater than $\bar{d}/\varepsilon$.) Examples include $(1,\varepsilon n)$-approximation for maximal/maximum matching~\cite{NO08:constant,YYI12:constant}, $(O(\log \bar{d}/\varepsilon),\varepsilon n)$-approximation for maximum dominating set~\cite{PR07:sublinear_distributed,NO08:constant}.

\section{Preliminaries}

Let $G=(V,E)$ be an undirected simple graph. We will assume that the vertex set $V$ of the
graph is $[n]:=\{1,\dots,n\}$ and that $n$ is known to the algorithm. We sometimes use lowercase letters (e.g., $u,v,x$) to denote the vertices, and in such cases, we let $\lab(u)\in [n]$ denote the index of the vertex $u$. We let $\Stream(G)$ denote the input stream of edges that define the graph $G$. We consider streaming
algorithm for random order streams, i.e. the input to our algorithm is a sequence of edges
that is chosen uniformly at random from the set of all permutations of the edges in $E$.
We are interested in algorithm that use constant or polylogarithmic space in the size of the
graph, where we count the number of words, i.e. if one is interested in the number of bits
used by the algorithm the space bounds have to be multiplied by a factor of $O(\log n)$.

\section{Approximating the Number of Connected Components (CCs)}

We start with an overview of the algorithm and the main ideas behind it. For a given value $k$ our algorithm will approximate the number of connected
components with exactly $k$ vertices with an additive error up to $\delta n$. In order to approximate the number of connected components with additive
error $\varepsilon n$ we approximate the number of small (of size at most $2/\varepsilon$) connected components by running our algorithm for each small value of 
$k$ such that the overall estimation error is at most $\varepsilon n/2$. Since the number of connected components of size more than $2/\varepsilon$ can be at 
most $\varepsilon n/2$ this gives the desired estimate.

By the considerations given in Section~\ref{sec:techniques}, in order to estimate the number of connected components of size $k$, we can use the following the approach: sample a sufficiently large number of vertices, keep only the sample vertices from which the greedy traversal discovers exactly $k$ vertices when processing the first $\lambda m$ edges, and check whether there
are outgoing edges in the remainder of the stream. However, there are challenges \ref{item:challenge_a} and \ref{item:challenge_b} for the analysis given in Section~\ref{sec:techniques}. We address challenge \ref{item:challenge_a} by using careful counting arguments and show that the probability of missing a witness for any false positive is sufficiently small. The second challenge~\ref{item:challenge_b} is more subtle: The probability of discovering $k$ vertices depends on the structure of the connected component. This implies that the set of vertices from which we discover exactly $k$ vertices is no longer uniformly distributed. This is one of the technical difficulties we need to overcome.


We address the issue by defining a canonical breadth first search (CBFS) and we accept a vertex only if this CBFS
discovers exactly $k$ vertices. From every vertex there is exactly one CBFS, which means that if we could verify
whether the $k$ vertices were discovered during a CBFS and accept only then we would still get a uniform distribution. 
The problem is that, similarly as in the previous discussion, there is a chance that we are discovering a set of edges 
that looks like a CBFS, but we have missed some edges that make our CBFS act differently. However, we can again overcome this 
problem by checking whether such an edge appears in the remainder of the stream. By similar arguments as above, it
turns out that the probability to have missed an edge conditioned on the event that our CBFS discovers $k$ vertices
is small. This implies that our distribution is almost uniform and so we can do the following approach:
\begin{itemize}
	\item
	Step 1: Sample a set $S$ of vertices.
	\item
	Step 2: Process the first $\lambda m$ edges of the stream and perform a CBFS from all vertices from $S$. If a vertex discovers
	more than $k$ vertices, then drop it as the component has size bigger than $k$.
	\item
	Step 3: Let $S^*$ be the set of vertices from which CBFS explored exactly $k$ vertices. For each $v\in S^*$ denote by $C_v$ the 
	connected component CBFS discovered so far.
	\item
	Step 4: Read the remainder of the stream. For each $v\in S^*$: If there is at least one outgoing edge from $C_v$ then set $X_v = 0$; else set $X_v=1$.
	\item
	Step 5: Extrapolate the number of connected component of size $k$ from $\sum_{v \in S^*} X_v$.
\end{itemize}

We also observe that in the above sketch we need to know the number of edges $m$ in advance. One can remove this assumption by guessing the number of edges upto a constant factor and having an instance for each guess, while this will blow up the space requirement by a $\Theta(\log n)$ factor. 
Another technical issue is that calculating the probabilities of being among the first $\lambda m$ edges can be cumbersome because under 
our conditioning it is not exactly $\lambda$. The following elegant modification of the algorithm will greatly simplify the analysis and resolve these two issues: Instead of
considering the prefix of $\lambda m$ edge we pick the length of the prefix from a suitable binomial distribution. This allows us to think of the prefix
as well as the random permutation defined by the following process: For each edge we choose a random priority from $[0,1]$. The random ordering
is then given by sorting the edges according to increasing priorities. Furthermore, the prefix is defined by all edges whose priorities 
are below a given threshold $t$. This is equivalent to choosing a random permutation and drawing the threshold from a binomial distribution. In addition, by storing the time-steps of each collected edge in the tree, we can defer the process of checking for false positives in the second phase to the very end of the stream, at which time we have already seen all the edges and chosen the length of the prefix. We remark that the idea of choosing the threshold
from a binomial distribution is somewhat similar to the high level idea of
\emph{Poissonization} technique that helps to break up dependencies between random variables (see e.g.,~\cite{Szp01:poisson,RRSS09:strong}).

In the following we give a formal analysis of our approach.

\subsection{Canonical Breadth First Search (CBFS) Tree}
For each vertex $v$, we define the \emph{Canonical Breadth First Search (CBFS) tree $\CT_k(v)$ of $v$ up to $k$ vertices} to be the tree with root $v$, that is constructed by performing a BFS starting from $v$ such that whenever the neighborhood of a vertex $u$ is to be explored, the neighbors of $u$ will be visited according to its lexicographical order, and the exploration stops if $k$ vertices have been reached or no more new vertices can be reached. The pseudocode for constructing such a tree is given in Algorithm~\ref{alg:bfs}.

\begin{algorithm}[t]
	\caption{CBFS tree of $v$ up to $k$ vertices}~\label{alg:bfs}
	\begin{algorithmic}[1]
		\Procedure{\textsc{CBFSTree}}{$G$,$v$,$k$}	
		\State queue $Q\gets\emptyset$; push $v$ to $Q$
		\State create a tree $T$ with root vertex $v$ 
		\While{$|V(T)|<k$ and $Q$ is not empty}
		\State pop a vertex $u$ from $Q$
		\State let $w_1,\cdots,w_{\deg(u)}$ be neighbors of $u$ that are ordered lexicographically
		\For{$i=1$ to $\deg(u)$}
		\If{$w_i$ is not in $T$}
		\State add vertex $w_i$ and edge $(u,w_i)$ to $T$
		\If{$|V(T)|=k$}
		\State \Return tree $T$
		\EndIf
		\State push $w_i$ to $Q$
		\EndIf
		\EndFor
		\EndWhile
		\State \Return tree $T$
		\EndProcedure
	\end{algorithmic}
\end{algorithm}  

We note that the way we construct a CBFS tree $\CT_k(v)$ (of $v$ up to $k$ vertices) also defines an ordering over all edges in $\CT_k(v)$. We will call such an edge ordering a \emph{CBFS ordering}.

\subsection{Detecting a Vertex Belonging to a Connected Component of Size $k$}
In order to detect if a vertex $v$ belongs to some connected component $C$ of size $k$, we will start from $v$ to collect edges from the stream. One naive way of doing this is to first collect every edge that is adjacent to $v$, and gradually collect edges that are incident to the edges collected so far until we see $k$ vertices. As we discussed before, in this way, the probability of observing a size $k$ component will be dependent on the structure of $C$, even if $C$ is promised to have exactly $k$ vertices. Instead, we collect edges (to form a tree) in a more restrictive way by trying to recover the CBFS tree $\CT_k(v)$. We first introduce the following definition. Let $T$ be a rooted tree and let $x$ be a vertex in $V(T)$. We let $\dep(x)$ denote the depth of $x$ in $T$, i.e., $\dep(x)$ is the length of the path from the root to $x$. Recall that each vertex $u\in V$ has a unique label $\lab(u)\in [n]$.

\begin{definition}~\label{def:violating_tree}
	Given a rooted tree $T$ and an edge $e=(u,v)$ such that $e \notin E(T)$, we call $e$ a \emph{violating edge} for $T$ if either 
	\begin{enumerate}
		\item exactly one of $u,v$, say $u$, belongs to $V(T)$, and there exists a vertex $y\in V(T)$ such that $\dep(y)-\dep(u)\geq 2$ or an edge $(u,x)\in E(T)$ such that $\lab(x)>\lab(v)$; or
		\item both $u,v$ belong to $V(T)$, and their depths $\dep(u),\dep(v)$ (i.e., distances to the root) are different, say $\dep(u)$ is smaller, and $\dep(v)-\dep(u)\geq 2$ or there exists an edge $(u,x)\in E(T)$ such that $\lab(x)>\lab(v)$.
	\end{enumerate}  
	Thus, $e$ is a witness that $T$ is not a CBFS tree.
\end{definition}


As mentioned before, we check if a vertex belongs to a CC of size $k$ by a two-phase procedure. More formally, we divide the edges in the random order stream into two phases such that any specific edge appears in the first phase of the stream with probability $\tau$, for some appropriately chosen parameter $\tau$. 

Now suppose we have already divided the stream into two phases with the desired property. For any input vertex $v$, we gradually collect the edges from the stream to form a tree $F$ rooted at $v$ as in the naive way in the first phase of the stream, while we stop collecting $F$ once we see a violating edge $e$ for the current tree or we have seen more than $k$ vertices and return \Bad. If we succeed in the first phase (i.e., neither of these two bad events happened), then in the second phase, we continue checking if there is some violating edge for the collected $F$ or there is an edge connecting $F$ to the rest of the graph. If either of these two bad events happens, then we return \Bad; otherwise, $v$ will be treated as a candidate that belongs to CC of size $k$.

Now we would like to define how we exactly determine the first part of the stream. For this purpose, consider the following way to generate a random permutation: For each edge $e$ we generate independently a random priority from the interval $[0,1]$. Then we order the edges according to increasing priorities. Now we define the prefix of the edges that belong to the first phase to be all edges whose priority is at most $\tau$. This way, we would have the nice property that the probability of an edge occuring in the first part of the stream is independent of the appearance of other edges. We further observe that the length of the first phase follows the binomial distribution $\Bi(m,\tau)$ and by symmetry the probability of seeing a particular permutation of the edges is independent of the length of the prefix. Thus, we can generate the same distribution by selecting a random permutation and
a random number $\Lambda$ from $\Bi(m,\tau)$ and defining the first part to consist of the first $\Lambda$ edges of the stream. By the above argument, for any edge the probability 
to belong to the first part of the stream is exactly $\tau$ and independent of the appearance of 
other edges. This property will greatly simplify our analysis.

Below we describe this algorithm in pseudocode (Algorithm~\ref{alg:stream}). 
\begin{algorithm}[H]
	\caption{Collecting $k$ neighbors of $v$ from the stream}~\label{alg:stream}
	\begin{algorithmic}[1]
		\Procedure{Stream\_CanoTree}{\Stream$(G)$,$v$,$k$,$\tau$,$\Lambda$}	
		\If{$k=1$}
		\State{\Return \Good}
		\EndIf
		\State create a tree $F$ with root vertex $v$
		\For{$(u,w)\gets$ next edge in the stream}
		\If{$(u,w)$ is a violating edge for $F$}
		\State{\Return \Bad}
		\ElsIf{exactly one of $u,w$ is contained in $V(F)$, say $u\in V(F)$}
		\State add vertex $w$ and edge $(u,w)$ to $F$
		\State record the time-step $t_{(u,w)}$ of $(u,w)$
		\If{$|V(F)|>k$}
		\State \Return \Bad \Comment{Large CC}
		\EndIf
		\EndIf
		\EndFor
		\If{$|V(F)|<k$}
		\State \Return \Bad  \Comment{Small CC}
		\Else \Comment{$|V(F)|=k$}
		\State $t_\ell\gets$ the time-step of the last edge added to $F$ 
		\If{$t_\ell > \Lambda$}
		\State \Return \Bad
		\Else
		\State \Return \Good
		\EndIf
		\EndIf
		\EndProcedure
	\end{algorithmic}
\end{algorithm}

Since we are storing for each tree the time-step of its last edge, we can simply choose $\Lambda$ at
the end of the stream, when we know the number of edges $m$.

Another important observation is that if the edges in the CBFS tree $\CT_k(v)$ of $v$ up to $k$ vertices appear in the CBFS ordering, i.e., the order how we define $\CT_k(v)$, then we will not encounter any violating edge when we are collecting these edges. That is, the CBFS tree will be collected given this edge ordering. Furthermore, if all these edges appear in the first phase of the stream, then Algorithm~\ref{alg:stream} will return \Good, and the resulting tree $F=\CT_k(v)$. In general, if the algorithm returns \Good, we call the corresponding tree $F$ a \emph{lexicographical breadth first search (LBFS) tree}, as the algorithm always collects edges in a lexicographical order (over the labels of corresponding vertices). Therefore, CBFS tree $\CT_k(v)$ is a special case of LBFS tree, and there might be other LBFS trees that are not the CBFS tree depending on the structure of the neighborhood of $v$ and how the edges are ordered in the stream.   


\subsection{Approximating the Number of Connected Components}
In the following, we present our algorithm for approximating the number of connected components.


We first introduce some notations. Let $k$ be any number such that $1\leq k\leq 2/\varepsilon$. Let $\gamma_k$ denote the probability that any set $T$ of $k-1$ edges appears in any specific order $\theta$ over $T$ in the first phase of the stream. Note that $\gamma_k$ is independent of the set $T$. Since there are $(k-1)!$ permutations for any $k-1$ edges, and each such edge appear in the first phase with probability $\tau$, we have that $\gamma_k=\frac{\tau^{k-1}}{(k-1)!}$. 
One important property we have is that the probability that the CBFS tree of $v$ is observed in the first phase of the stream is exactly $\gamma_k$.


Our algorithm is to approximate the number of connected components of size $k$ for each $1\leq k\leq 2/\varepsilon$. For that, we sample a constant number of vertices, and from each sampled vertex, we use \textsc{Stream\_CanoTree} to detect if it is a witness of some connected component of size $k$. We then define the appropriate estimator from these witnesses. 
%
%
%
%
The details are described in Algorithm~\ref{alg:approx_num_cc}.
\begin{algorithm}[H]
	\caption{Approximating the number of CCs}~\label{alg:approx_num_cc}
	\begin{algorithmic}[1]
		\Procedure{NumCC}{\Stream$(G)$,$\varepsilon$,$\rho$}	
		\State $\tau\gets \Theta((4/\varepsilon)^{-6/\varepsilon^2-3}\rho)$ 
		\State choose a number $\Lambda$ from $\Bi(m,\tau)$
		\State sample a set $A$ of $s:=\Theta((4/\varepsilon)^{15/\varepsilon^3}(1/\rho)^{2/\varepsilon+1})$ vertices uniformly at random
		\For{each $k$ from $1$ to $2/\varepsilon$}
		\For{each $v\in A$}
		\If{\textsc{Stream\_CanoTree}(\Stream($G$),$v$,$k$,$\tau$,$\Lambda$) returns \Good}
		
		\State $X_v=1$
		\Else
		\State{$X_v=0$}
		\EndIf
		\EndFor
		\State $\gamma_k\gets\frac{1}{(k-1)!}\cdot \tau^{k-1}$
		\State $C_k=\frac{\sum_{v\in A}X_v}{|A|}\cdot \frac{n}{k}\cdot \frac{1}{\gamma_k}$
		\EndFor
		\State \Return $\sum_{k=1}^{2/\varepsilon}C_k$. 
		\EndProcedure
	\end{algorithmic}
\end{algorithm}

%

We remark that in the above algorithm, sampling the number $\Lambda$ from $\Bi(m,\tau)$ can be implemented in constant space as follows. Initially, we set a counter $c$ to be $0$. Next for each edge from the stream, we flip a biased coin with HEAD probability $\tau$, and increase the counter $c$ by $1$ if a HEAD is seen. It is not hard to see that the final counter $c$ has the same probability distribution as $\Bi(m,\tau)$.

Our theorem is as follows.
\begin{theorem}~\label{thm:num_cc}
	Let $1/2 > \varepsilon,\rho>0$. The algorithm \textsc{NumCC} takes as input any graph $G$ from a random order stream, $\varepsilon,\rho$, approximates the number of connected components of $G$ within additive error $\varepsilon n$,  with probability at least $1-\rho$. The space complexity of this algorithm is $O((1/\varepsilon)^{O(1/\varepsilon^3)}(1/\rho)^{O(1/\varepsilon)})$. Furthermore, each update can be processed in  $O((1/\varepsilon)^{O(1/\varepsilon^3)}(1/\rho)^{O(1/\varepsilon)})$ time and the output can be computed in $O((1/\varepsilon)^{O(1/\varepsilon^3)}(1/\rho)^{O(1/\varepsilon)})$ post-processing time.
\end{theorem}



\subsection{Proof of Theorem~\ref{thm:num_cc}}
Now we present the analysis of the algorithm and prove Theorem~\ref{thm:num_cc}. First note that, besides the usage of $O(1)$ space to simulate the sampling $\Lambda$ from a Binomial distribution, the algorithm \textsc{NumCC} only needs to store for each sampled vertex in $A$, a corresponding tree with size up to $O(1/\varepsilon)$ vertices. This implies that the space used by \textsc{NumCC} is at most $O(|A|/\varepsilon)=O((4/\varepsilon)^{15/\varepsilon^3}(1/\rho)^{2/\varepsilon+1})$. Note that for each update, we only need to check for each $k\leq 2/\varepsilon$, if the edge is violating for each of collected trees or witnesses a large component (of size larger than $k$). Since each such tree has size at most $k\leq 2/\varepsilon$ and the total number of collected trees at any time is $O(|A|)$, we know that each update can be processed in time $O(|A|/\varepsilon^2)=O((4/\varepsilon)^{15/\varepsilon^3+2}(1/\rho)^{2/\varepsilon+1})$. To compute the final estimator, it suffices to have the values $X_v$ for each vertex $v\in A$ in each iteration. Thus, the estimator can be computed in $O(|A|/\varepsilon)=O((4/\varepsilon)^{15/\varepsilon^3+1}(1/\rho)^{2/\varepsilon+1})$ time in the post-processing procedure.  
In the following, we prove the correctness of the algorithm \textsc{NumCC}.

For any vertex $v$, we let $\mathbf{C}(v)$ denote the connected component containing $v$. When it is clear from the context, we also let $\mathbf{C}(v)$ denote the set of vertices in the connected component containing $v$. We let $\cc_k$ denote the number of CCs with size $k$. We let $\cc_{\leq k}$ denote the number of CCs with size at most $k$. For any two vertices $u,v$, we let $\dt_G(u,v)$ denote the distance between $u,v$ in the graph $G$ that is defined by the edge stream.

We write $A\sim \calU_V$ to indicate that $A$ is a set of $s$ vertices sampled uniformly at random from $V$. 
Let $\calP[E]$ denote the set of all permutations of edges. We write $\sigma \sim\calU_{\calP[E]}$ to indicate that an edge ordering $\sigma$ is sampled from the uniform distribution over $\calP[E]$. 

For any integer $k\geq 1$, we let $A_k\subseteq A$ denote the subset of vertices in $A$ that belong to a CC of size $k$, that is, $A_k:=\{v|v\in A, |\mathbf{C}(v)|=k\}$. By Chernoff bound and our setting that $s=\Theta((4/\varepsilon)^{15/\varepsilon^3}(1/\rho)^{2/\varepsilon+1})$, we have the following lemma regarding the sample set $A$. 

\begin{lemma}\label{lemma:property_A}
	With probability (over the randomness of $A\sim \calU_V$) at least $1-\frac{\rho}{5}$, it holds that for each $k\leq 2/\varepsilon$, $|\frac{|A_k|}{|A|}-\frac{k\cdot \cc_k}{n}|\leq \frac{\varepsilon^2}{16}$ and no two vertices in $A$ belong to the same connected component of size $k$.
\end{lemma}

In the following, we will fix $A$ and condition on the above  events listed in Lemma~\ref{lemma:property_A}, which hold with probability at least $1-\frac{\rho}{5}$. 

Now let $k$ be an integer such that $1\leq k\leq \frac{2}{\varepsilon}$, and we analyze the estimator $C_k=\frac{\sum_{v\in A}X_v}{|A|}\cdot \frac{n}{k}\cdot \frac{1}{\gamma_k}$, where $X_v$ is the indicator variable of the event that \textsc{Stream\_CanoTree}($\Stream(G),v,k,\tau$) returns \Good. 

Let $v\in A$. We distinguish the following three cases according to the size of the connected component containing $v$.

\textbf{(\rom{1}) $|\mathbf{C}(v)|<k$.} In this case, \textsc{Stream\_CanoTree} will always return \Bad~as the maximum number of vertices in the collected tree is at most $k-1$. Thus, with probability $1$, it holds that $X_v=0$.

\textbf{(\rom{2}) $|\mathbf{C}(v)|=k$.} Note that $X_v=1$ will happen if either 1) the CBFS tree $\CT_k(v)$ of $v$ up to $k$ vertices is observed in the lexicographical order in the first part of the stream, or 2) some other lexicographical tree $\LT_k(v)$ other than $\CT_k(v)$ is observed in the first part of the stream and the edges in $\CT_k(v)-\LT_k(v)$ appeared in the ``non-collectable'' order in the first part of the stream, that is, no edge in $\CT_k(v)-\LT_k(v)$	was detected as a violating edge by \textsc{Stream\_CanoTree} and \textsc{Stream\_CanoTree} returns \textbf{Good} and collects $\LT_k(v)$ at the end of the stream. The first event happens will probability $\gamma_k$. 

Now let us bound the probability of the second event. Note that for any  lexicographical tree $\LT_k(v)$, the probability that $\LT_k(v)$ is observed in the first phase of the stream is $\gamma_k$. Conditioned on that, the probability that all the edges in $\CT_k(v)-\LT_k(v)$ appeared in the ``non-collectable'' order in the first part of the stream is at most $\tau^{|E(\CT_k(v)-\LT_k(v))|}$, which is at most $\tau$, as $|E(\CT_k(v)-\LT_k(v))|\geq 1$. Since there are at most $\binom{k(k-1)/2}{k-1}\leq k^{2(k-1)}$ such lexicographical trees, the probability that the second event happens is at most $k^{2(k-1)}\cdot \tau\cdot\gamma_k$.


Therefore, it holds that 
\begin{eqnarray*}
\gamma_k&\leq& \Pr_{\sigma\sim\calU_{\calP[E]}}[X_v=1] 
\leq
\gamma_k + k^{2(k-1)}\cdot\tau\cdot \gamma_k 
\leq
(1+\varepsilon)\gamma_k.
\end{eqnarray*}

\textbf{(\rom{3}) $|\mathbf{C}(v)|>k$.} For any vertex $v$ that belongs to a CC of size $s\geq k$, there exist a unique CBFS tree $\CT_k(v)$; there might be multiple LBFS trees $\LT_k(v)$. Note that $X_v=1$ is equivalent to that either 1) $\CT_k(v)$ is observed in the first part of the stream, and all the edges between $\CT_k(v)$ and the remaining part of the graph have appeared in the ``non-collectable'' order in the first phase of the stream; 2) some lexicographical tree $\LT_k(v)$ other than $\CT_k(v)$ is observed in the first phase of the stream and the edges in $\CT_k(v)-\LT_k(v)$ and $E(\LT_k(v),V - \LT_k(v))$ appeared in the ``non-collectable'' order in the first phase of the stream. In the following, we bound the probability of these events. Our main lemma is as follows.

\begin{lemma}~\label{lemma:large_CC_X_v}
	Let $\eta=\Theta(\rho\varepsilon^3)$. Let $v\in A$ be a vertex that belongs to a CC of size $s> k$. Then it holds that $\Pr_{\sigma\sim\calU_{\calP[E]}}[X_v=1]\leq \eta\cdot \gamma_k$.
\end{lemma}

To prove the above lemma, we introduce some definitions.
\begin{definition}
	Let $\calT_k(v)$ denote the set of all LBFS trees rooted at $v$. Let $\calT_{k,b}(v)$ denote the set of LBFS trees $T\in \calT_k(v)$ with boundary size $b$, i.e., $|E(V(T), \mathbf{C}(v)\setminus V(T))|=b$. Let $t_{k,b}(v)=|\calT_{k,b}(v)|$.
\end{definition}
Note that for any $v$ with $|\mathbf{C}(v)|= k$, it is apparent that all the LBFS trees rooted at $v$ of size $k$ have no connections to the rest of the graph, that is, $t_{k,b}(v)=0$ for any $b\geq 1$. For any $v$ with $|\mathbf{C}(v)|> k$, and any LBFS tree $T$ rooted at $v$ of size $k$, there must be at least one edge lie between $V(T)$ and $\mathbf{C}(v)\setminus V(T)$, that is, $t_{k,0}(v)=0$.

\begin{lemma}~\label{lemma:tkb_upper}
	Let $v$ be a vertex that belongs to a CC of size at least $k$. Then it holds that for any $b\geq 1$,  $t_{k,b}(v)\leq k^{2(k-1)}\cdot(k-1+b)^{k(k-1)}$. 
\end{lemma}
\begin{proof}
	Note that for any LBFS tree $T\in \calT_{k,b}(v)$, only vertices with degree at most $k-1+b$ can be contained in $T$, as otherwise the number of edges between $T$ and $C-T$ is at least $b+1$, which contradicts the definition of $T$. In addition, only vertices that are within distance at most $k-1$ from $v$ can be contained in $T$. Since the total number of vertices that are within distance at most $k-1$ from $v$ and with degree at most $k-1+b$ is at most $(k-1+b)^{k}$, this also implies the same upper bound on the number of vertices that can be contained in any of tree in $\calT_{k,b}(v)$.
	
	Finally, as the total number of different LBFS rooted tree of size $k$ in a graph with at most $\lambda$ vertices is at most $\binom{\lambda-1}{k-1}\cdot \binom{k^2}{k-1}\leq (k^2\lambda)^{k-1}$, we have that $t_{k,b}(v)\leq k^{2(k-1)}\cdot(k-1+b)^{k(k-1)}$.
\end{proof}

Now we are ready to prove Lemma~\ref{lemma:large_CC_X_v}.
\begin{proof}[Proof of Lemma~\ref{lemma:large_CC_X_v}]
	For any $T\in \calT_k(v)$, let $\calE_T$ denote the event that $T$ is observed from $v$, all edges in $E(\CT_k(v)-T)$ and $E(T,V - T)$ appeared in the ``non-collectable'' order in the first phase of the stream. Let $\calE_T'$ denote the event that all edges in $E(\CT_k(v)-T)$ and $E(T,V - T)$ appeared in the first phase of the stream. For any $b\geq 1$ and $T\in \calT_{k,b}(v)$, we note that $|E(\CT_k(v)-T)\cup E(T,V - T)|\geq b$, and thus
	\begin{eqnarray*}
		\Pr_{\sigma\sim\calU_{\calP[E]}}[\calE_T]
		\leq
		\Pr_{\sigma\sim\calU_{\calP[E]}}[\textrm{$T$ is observed}]\cdot \Pr_{\sigma\sim\calU_{\calP[E]}}[\calE_T']
		\leq 
		\gamma_k\cdot \tau^{b}
	\end{eqnarray*}
	
	Then we have that 
	\begin{eqnarray*}
		\Pr_{\sigma\sim\calU_{\calP[E]}}[X_v=1] 
		&=& \Pr_{\sigma\sim\calU_{\calP[E]}}[\bigcup_{T\in \calT_k(v)} \calE_T] \\
		&\leq& \sum_{T\in \calT_k(v)} \Pr_{\sigma\sim\calU_{\calP[E]}}[\calE_T]\\
		&=&  \sum_{b\geq 1} \sum_{T \in \calT_{k,b}(v)} \Pr_{\sigma\sim\calU_{\calP[E]}}[\calE_T]\\
		&\leq& \sum_{b\geq 1} t_{k,b}(v)\cdot \gamma_k \cdot \tau^{b}\\
		&\leq& \sum_{b \geq 1} k^{2(k-1)}\cdot(k-1+b)^{k(k-1)} \cdot \gamma_k \cdot \tau^{b}\\
		&\leq& k\cdot k^{2(k-1)}\cdot (2k)^{k(k-1)}\cdot \gamma_k\cdot \tau\\
		&&+ \sum_{b \geq k} k^{2(k-1)}\cdot(2b)^{k(k-1)} \cdot \gamma_k \cdot \tau^{b}\\
		&\leq &(2k)^{k^2+k}\tau\cdot\gamma_k+(2k)^{k^2}\cdot \sum_{b\geq k} e^{k(k-1)\ln b-\ln\frac{1}{\tau} \cdot b}\gamma_k\\
		&\leq& (2k)^{k^2+k}\tau\cdot\gamma_k+(2k)^{k^2}\cdot \sum_{b\geq k} \tau^{2b}\gamma_k\\
		&\leq& \eta\cdot\gamma_k,
	\end{eqnarray*}
	where the second inequality follows from Lemma \ref{lemma:tkb_upper} and where in the last two inequalities we used the facts that $k\leq \frac{2}{\varepsilon}$, $\tau=\Theta((4/\varepsilon)^{-6/\varepsilon^2-3}\rho)$, $\gamma_k=\frac{1}{(k-1)!}\cdot\tau^{k-1}$ and $\eta=\Theta(\rho\varepsilon^3)$.
\end{proof}

Therefore, we have that $\gamma_k\leq \E_{\sigma\sim\calU_{\calP[E]}}[X_v]\leq (1+\varepsilon)\gamma_k$ if $|\mathbf{C}(v)|=k$; and that $\E_{\sigma\sim\calU_{\calP[E]}}[X_v]\leq \eta\cdot \gamma_k$ if $|\mathbf{C}(v)|>k$. 

Now let $Y_1:=\sum_{v\in A_k} X_v$. It holds that $\E_{\sigma\sim\calU_{\calP[E]}}[Y_1]=\sum_{v\in A_k} \E_{\sigma\sim\calU_{\calP[E]}}[X_v]$ and thus that $|A_k|\gamma_k\leq \E_{\sigma\sim\calU_{\calP[E]}}[Y_1]\leq (1+\varepsilon)|A_k|\gamma_k$. Furthermore, since we have conditioned on the event that no two vertices in $A$ belong to the same connected component of size $k$, the random variables in $\{X_v:v\in A_k\}$ are independent of each other. Thus, $\Var_{\sigma\sim\calU_{\calP[E]}}[Y_1]=\sum_{v\in A_k}\Var_{\sigma\sim\calU_{\calP[E]}}[X_v]\leq |A_k|\cdot \E_{\sigma\sim\calU_{\calP[E]}}[X_v]\leq |A_k|\cdot(1+\varepsilon)\gamma_k$. Thus, by Chebyshev inequality, 
\begin{eqnarray*}
\Pr_{\sigma\sim\calU_{\calP[E]}}[|Y_1-\E[Y_1]|
\geq
\frac{\varepsilon^2 \gamma_k|A|}{16}]
\leq \frac{256\cdot(1+\varepsilon)\gamma_k|A_k|}{\varepsilon^4 |A|^2\gamma_k^2} 
\leq \frac{\varepsilon\rho}{20}, 
\end{eqnarray*}
where in the last inequality, we used the fact that $|A_k|\leq |A|$ and that $|A|=\Theta((4/\varepsilon)^{15/\varepsilon^3}(1/\rho)^{2/\varepsilon+1})=\Omega(\frac{1}{\varepsilon^5\gamma_k\rho})$.
This further implies that with probability at least $1-\frac{\varepsilon\rho}{20}$, 
$$|A_k|\gamma_k - \frac{\varepsilon^2 \gamma_k|A|}{16}\leq Y_1\leq(1+\varepsilon)|A_k|\gamma_k +\frac{\varepsilon^2 \gamma_k|A|}{16}.$$

Now let $Y_2:=\sum_{j\geq k+1}\sum_{v\in A_j} X_v$. Then it holds that $$\E_{\sigma\sim\calU_{\calP[E]}}[Y_2]=\E_{\sigma\sim\calU_{\calP[E]}}[\sum_{j\geq k+1}\sum_{v\in A_j}X_v]\leq \sum_{j\geq k+1}|A_j|\cdot \eta\cdot \gamma_k\leq |A|\eta\gamma_k.$$

By Markov's inequality, we have that 
\begin{eqnarray*}
	\Pr_{\sigma\sim\calU_{\calP[E]}}[Y_2\geq \frac{\varepsilon^2 |A|\gamma_k}{16}] \leq \frac{16\cdot \E[Y_2]}{\varepsilon^2 |A|\gamma_k} \leq \frac{\varepsilon\rho}{20},
\end{eqnarray*}
where in the last inequality we used the fact that $\eta=O(\rho\varepsilon^3)$.

Therefore, with probability at least $1-\frac{\varepsilon\rho}{10}$, 
\begin{eqnarray*}
	|A_k|\gamma_k-\frac{\varepsilon^2 |A|\gamma_k}{16}\leq Y_1+Y_2\leq (1+\varepsilon)|A_k|\gamma_k+\frac{\varepsilon^2 |A|\gamma_k}{8}
\end{eqnarray*}

Now since $C_k=\frac{\sum_{v\in A}X_v}{|A|}\frac{n}{k}\frac{1}{\gamma_k}=\frac{Y_1+Y_2}{|A|}\frac{n}{k}\frac{1}{\gamma_k}$, we have that with probability at least $1-\frac{\varepsilon\rho}{10}$, 
\begin{eqnarray*}
C_k &\geq& \frac{|A_k|\gamma_k-\frac{\varepsilon^2 |A|\gamma_k}{16}}{|A|}\frac{n}{k\gamma_k}
\geq \left(\frac{k\cdot\cc_k\cdot \gamma_k}{n}-\frac{\varepsilon^2}{16}\gamma_k-\frac{\varepsilon^2 \gamma_k}{16}\right)\frac{n}{k\gamma_k}
\geq \cc_k-\frac{1}{8}\varepsilon^2n,
\end{eqnarray*}
 and 
\begin{eqnarray*}
	C_k\leq \frac{(1+\varepsilon)|A_k|\gamma_k+\frac{\varepsilon^2 |A|\gamma_k}{8}}{|A|}\frac{n}{k}\frac{1}{\gamma_k}
	&\leq& \left(\frac{(1+\varepsilon)\cdot k\cdot\cc_k\cdot \gamma_k}{n}+\frac{\varepsilon^2}{16}\gamma_k+\frac{\varepsilon^2 \gamma_k}{8}\right)\frac{n}{k\gamma_k}\\
	&\leq& \cc_k+\frac{1}{4}\varepsilon^2n,
\end{eqnarray*}
where in the above inequalities, we used the condition that $|\frac{|A_k|}{|A|}-\frac{k\cdot\cc_k}{n}|\leq \frac{\varepsilon^2}{16}$.

Therefore, with probability at least $1-\frac{\varepsilon\rho}{10}\cdot \frac{2}{\varepsilon}=1-\frac{\rho}{5}$, for all $k\leq \frac{2}{\varepsilon}$, $\cc_k-\frac{1}{8}\varepsilon^2 n\leq C_k\leq \cc_k + \frac{1}{4}\varepsilon^2 n$. This further implies that
$$\sum_{k=1}^{2/\varepsilon}\cc_k-\sum_{k=1}^{2/\varepsilon}\frac{1}{8}\varepsilon^2 n\leq \sum_{k=1}^{2/\varepsilon}C_k\leq \sum_{k=1}^{2/\varepsilon}\cc_k +\sum_{k=1}^{2/\varepsilon}\frac{1}{4}\varepsilon^2 n.$$
That is,
$$\cc_{\leq 2/\varepsilon}-\frac{1}{4}\varepsilon n\leq \sum_{k=1}^{2/\varepsilon}C_k\leq \cc_{\leq 2/\varepsilon} +\frac{1}{2}\varepsilon n.$$

The statement of the theorem then follows by taking the union bound and the fact that the total number of connected components with size at least $2/\varepsilon$ is at most $\varepsilon n/2$.

\section{Approximating the Weight of the Minimum Spanning Tree}

We now consider the problem of estimating the weight of the minimum spanning tree in random orer streams. 
The algorithm presented in this section as well as its analysis are given in \cite{CRT05:MST}. 
We only describe them for sake of completeness.

We assume that our input graph $G=(V,E)$ has edge weights $w_e \in \{1,\dots,W\}$ for
each edge $e\in E$. The edges together with their weights appear in random order. In order to
approximate the weight of the minimum spanning tree we use the following representation from 
\cite{CRT05:MST}. For each $t\in\{1,\dots, W-1\}$ we define the threshold graph $G^{(t)} = (V,E^{(t)})$
to be the graph that contains only edges of edge weight at most $t$, i.e. $E^{(t)} = \{e \in E : 
w_e \le t\}$. Let $c(t)$ be the number of connected components in $G^{(t)}$. Then we can use the
following formula from \cite{CRT05:MST} to express the weight of the minimum spanning tree $M$:
$$
M = n - W + \sum_{t=1}^{W-1} c(t) .
$$
We will use the algorithm from the previous section to obtain an estimate $\hat c(t)$ for each
of the $c(t)$ in order to obtain an estimate 
$$
\hat M = n - W + \sum_{t=1}^{W-1} \hat c(t)
$$
for the weight of the minimum spanning tree. We remark that $M$ is at least $n-1 \ge n/2$ for $n\ge 2$.
Thus, if we approximate each of the $\hat c(t)$ within an additive error of $\pm \varepsilon n/(4W)$ then 
we have an overall additive error of $\varepsilon n/2$, which translates to
$$
(1-\varepsilon) \cdot M \le \hat M \le (1+\varepsilon) M .
$$

Below we give the algorithm from \cite{CRT05:MST} for approximating MST weight in pseudocode.

\begin{algorithm}[H]
	\caption{Approximating the weight of MST}~\label{alg:MST_weight}
	\begin{algorithmic}[1]
		\Procedure{MSTWeight}{\Stream$(G)$,$W$,$\varepsilon,\rho$}	
		\For{$t$ from $1$ to $W-1$}
		\State $\hat{c}(t)\gets$ \textsc{NumCC}(\Stream$(G^{(t)})$,$\frac{\varepsilon}{4W}$, $\frac{\rho}{W}$)
		\EndFor
		\State $n\gets$ the number of vertices of $G$
		\State \Return $\hat M = n - W + \sum_{t=1}^{W-1} \hat c(t)$
		\EndProcedure
	\end{algorithmic}
\end{algorithm}

By Theorem~\ref{thm:num_cc}, the above discussion and the union bound, we have the following theorem.
\begin{theorem}
	Let $0 < \varepsilon < 1/2$. Given the edges of a connected graph $G=(V,E)$ together with their weights from $\{1,\dots, W\}$ in a random order, algorithm \textsc{MSTWeight} computes with probability at least $1-\rho$ an estimate $\hat M$ for the weight of the minimum spanning tree of $G$ that satisfies
	$$
	(1-\varepsilon) \cdot M \le \hat M \le (1+\varepsilon) M .
	$$
	The algorithm uses space $O((\frac{W}{\varepsilon})^{O(W^3 /\varepsilon^3)}(W/\rho)^{O(W/\varepsilon)})$, processes each update in $O((W/\varepsilon)^{O(W^3 /\varepsilon^3)} \cdot (W/\rho)^{O(W/\varepsilon)})$ time and computes the estimate in 
	$O((W/\varepsilon)^{O(W^3/\varepsilon^3)}(W/\rho)^{O(W/\varepsilon)})$ post-processing time.
\end{theorem}

%

\section{Approximating the Size of Maximum Independent Set (MIS) in Minor-Free Graphs and Beyond}\label{sec:mis_planar}
In this section, we transform the constant-time approximation algorithm for the size of maximum independent set of minor-free graphs into a constant-space random order streaming algorithm, and prove Theorem~\ref{thm:planar_mis}. Our transformation can also be extended to many other constant-time approximation algorithms.
We first introduce the following definitions on $d$-bounded $k$-discs. 
\begin{definition}
	Let $G=(V,E)$ be an undirected graph. Let $d\geq 1$ and $k\geq 0$. Let $G_{|d}=(V,E')$ be the subgraph of $G$ that is obtained by removing all edges that are incident to vertices with degree higher than $d$. Let $v\in V$ be a vertex. 
	The $d$-bounded $k$-disc of a vertex $v$, denoted by $\disc_{k,d}(v,G)$, is defined to be the subgraph rooted at $v$ and induced by all vertices that are within distance at most $k$ from $v$ in $G_{|d}$. Let $N=N_{d,k}$ denote the number of all possible $d$-bounded $k$-disc isomorphism types\footnote{We call two rooted graphs isomorphic to each other if there is a root-preserving isomorphic mapping from one graph to the other.}. Let $\calH_{k,d}=\{\Delta_1,\cdots,\Delta_N\}$ denote the set of all such types. 
\end{definition}
\newcommand{\calJ}{\mathcal{J}}

Note that if both $d,k$ are constant, then $N=N_{d,k}$ is also constant. For any graph $G$ and a $d$-bounded $k$-disc type $\Delta$, the frequency of $\Delta$ with respect to $G$ is defined to be the number of vertices in $G$ with $d$-bounded $k$-disc isomorphic to $\Delta$.

Now let us very briefly describe the constant-time approximation algorithm for the size of maximum independent set in minor-free graphs, with an additive error~$\varepsilon n$~\cite{HKNO09:local}. 
The idea is that one can first ``ignore'' all the edges incident to vertices with degree higher than $d=\Theta(1/\varepsilon)$ (since the total number of such vertices is $O(\varepsilon n)$, the size of MIS changes by at most $O(\varepsilon n)$). To check if a vertex $v$ has degree higher than $d$, one only needs to perform $O(d)$ queries to its adjacency list. Then in the resulting bounded degree graph $G_{|d}$, one could 1) sample $O(1/\varepsilon^2)$ vertices; 2) for each sampled vertex $v$, perform a query to a \emph{local partition oracle} $\mathcal{O}$, which returns the component $P(v)$, of some global partition $\calP_{G_{|d}}$, containing $v$. It is guaranteed that the oracle $\mathcal{O}$ only needs to perform $q= d^{O(\log^2(1/\varepsilon))}=(1/\varepsilon)^{O(\log^2(1/\varepsilon))}$ queries to the adjacency list of the graph; and at most $\varepsilon n/3$ edges lie between different parts of the partitioning (see~\cite{HKNO09:local,LR15:quasi} for details on local graph partition oracle); 3) compute the fraction $f$ of sampled vertices $v$ that belongs to the optimal independent set for $P(v)$ and return $\bar{f}:=f\cdot n$ as the estimator for the size of MIS of $G$. It can be shown that $\bar{f}$ approximates the size of MIS of $G$ with additive error $\varepsilon n$ \cite{HKNO09:local}.

A key observation is that for the above algorithm (actually, most existing constant-query approximation algorithms in the adjacency list model), the decisions are independent of the labeling of the vertices of $G$ and only depend on the structure of the corresponding $k$-discs of sampled vertices (and the random bits used by each vertex inside the $k$-disc). Then to simulate these algorithms, 
it suffices to have the distribution of $d$-bounded $k$-disc types of the corresponding graph for $k=\Theta(q)$. This is true since the above algorithm can be equivalently seen as first sampling a constant number of $k$-disc types with probability proportional to their frequencies, and then invoke the local graph partition oracle on the root inside the corresponding $k$-discs. 
More formally, let $G$ be a minor-free graph. For each $d$-bounded $k$-disc type $\Delta$, we let $\cc_\Delta$ denote the number of vertices $v$ with $\disc_{d,k}(v,G)\cong \Delta$. Suppose that for each $\Delta\in \calH_{d,k}$, we have some estimator $C_\Delta$ such that $|C_\Delta -\cc_{\Delta}|\leq \rho n$ for some small constant $\rho>0$. Now instead of sampling vertices from the original graph $G$, we sample a number of $d$-bounded $k$-disc types with probability proportional to their estimators. Then we invoke the local partition oracle on each sampled disc type $\Delta$ by taking the root $r$ as input. Finally, we can still utilize the information whether $r$ belongs to the MIS or not to obtain the same estimator as before. If our estimator $C_\Delta$'s are good enough, that is, $\rho$ is small enough constant, then we can still approximate the size of MIS of $G$ with additive error $\varepsilon n$.

In the following, we show how to approximate the frequency of $d$-bounded $k$-disc types in a graph $G$ that is defined by a random order edge stream. Our goal is to obtain for any $d$-bounded $k$-disc type $\Delta$ an estimator $C_\Delta$ such that $|C_\Delta -\cc_{\Delta}|\leq \rho n$ for any small constant $\rho>0$ by using constant space that only depends on $\rho,d,k$, and not on the size of $G$. 

\subsection{Canonical Extended $(d+1)$-Bounded $k$-Disc}
Since we do not know the degree of each vertex until we read all the edges from the stream, we cannot directly approximate the $\cc_\Delta$ in ``truncated'' graph $G_{|d}$, which is trivial in the query access model~\cite{HKNO09:local} while it is a major technical difficulty in the streaming setting. In the following, instead of directly considering the $k$-disc of $v$ in $G_{|d}$, we consider a slightly larger neighborhood of $v$ in $G$, which can be further used to approximate the frequency of $\Delta$.

More specifically, let $v\in G$ be any vertex with degree at most $d$.  We define the \emph{extended $(d+1)$-bounded $k$-disc of $v$} to be the subgraph that includes $\disc_{k,d}(v,G)$, and also some vertices with degree higher than $d$ in $G$ and a subset of $d+1$ edges for each such vertex. The \emph{canonical} extended $(d+1)$-bounded $k$-disc of $v$, denoted by $\CD_{k,d+1}(v,G)$, is one such disc that edges and vertices are added sequentially and in each time according to some lexicographical ordering over the labels of vertices. The formal definition of $\CD_{k,d+1}(v,G)$ is given in Algorithm~\ref{alg:bounded_disc}, in which we let $\radius(F)$ denote the maximum distance from any vertex in a rooted graph $F$ to its root $r$. 

\begin{algorithm}[H]
	\caption{Canonical extended $(d+1)$-bounded $k$-disc of $v$}~\label{alg:bounded_disc}
	\begin{algorithmic}[1]
		\Procedure{\textsc{CanoDisc}}{$G$,$v$,$k$,$d$}
		\If{$k=0$}
		\State \Return $F:=\{v\}$
		\EndIf	
		\State queue $Q\gets\emptyset$; push $v$ to $Q$
		\State create a graph $F$ with root vertex $v$ 
		\While{$Q$ is not empty}
		\State pop a vertex $u$ from $Q$
		\State let $w_1,\cdots,w_{\deg(u)}$ be neighbors of $u$ that are ordered lexicographically
		\For{$i=1$ to $\min\{\deg(u),d+1\}$}
		\If{$\radius(F+\{u,w_i\})\le k$}
		\State add vertex $w_i$ if $w_i\notin{V(F)}$
		\State add edge $(u,w_i)$ to $F$
		\State push $w_i$ to $Q$
		\EndIf
		\EndFor
		\EndWhile
		\State \Return $F$
		\EndProcedure
	\end{algorithmic}
\end{algorithm}   
We now define the canonical $\CD_{k,d+1}(v,G)$-edge ordering to be the order that the edges are added into the rooted graph $F$ in Algorithm~\ref{alg:bounded_disc}. 

Note that in order to estimate $\cc_\Delta$, it suffices to approximate the number $\cc_\Gamma$ of each extended $(d+1)$-Bounded $k$-Disc type $\Gamma$ with small additive error.

In the following, we will give algorithm for estimating $\cc_\Gamma$ with small additive error. When it is clear from the context, we will omit $G$ in the notations $\disc_{k,d}(v,G)$ and $\CD_{k,d+1}(v,G)$.

\subsection{Detecting a Vertex with Extended $(d+1)$-Bounded $k$-Disc Isomorphic to $\Gamma$}
Similar as the approach for estimating the number of connected component, we obtain our estimator $C_\Gamma$ by using the statistics of the fraction of sampled vertices which serve as a candidate for an extended $(d+1)$-bounded $k$-disc isomorphic to $\Gamma$. 

Starting from a vertex $v$, we collect edges in its $k$-disc from the stream as follows. We first create a rooted graph $F$ with a single root vertex $v$. Then we gradually add to $F$ new edges and vertices according to the lexicographical order, and for each vertex, we will add to $F$ at most $d+1$ adjacent edges and we stop until no more vertices or edges that are within distance $k$ from $v$ can be found. We need the following definition of violating edges which is almost the same as Definition~\ref{def:violating_tree} except here we are considering a rooted subgraph rather than just a rooted tree. We let $\dep(x)$ denote the length of the shortest path from the root of $F$ to a vertex $x$.

\begin{definition}
	Given a rooted $d$-bounded graph $F$ and an edge $e=(u,v)$ such that $e \notin E(F)$, we call $e$ a \emph{violating edge} for $F$ if either 
	\begin{enumerate}
		\item exactly one of $u,v$, say $u$, belongs to $V(F)$, and there exists a vertex $y\in V(F)$ such that $\dep(y)-\dep(u)\geq 2$ or an edge $(u,x)\in E(F)$ such that $\lab(x)>\lab(v)$; or
		\item both $u,v$ belong to $V(F)$, and their depths $\dep(u),\dep(v)$ are different, say $\dep(u)$ is smaller, and $\dep(v)-\dep(u)\geq 2$ or there exists an edge $(u,x)\in E(F)$ such that $\lab(x)>\lab(v)$.
	\end{enumerate}  
\end{definition}

The procedure of collecting $k$-discs is formally defined in Algorithm~\ref{alg:stream_k_disc}.

\begin{algorithm}[H]
	\caption{Collecting $k$-disc of $v$ from the stream}~\label{alg:stream_k_disc}
	\begin{algorithmic}[1]
		\Procedure{Stream\_CanoDisc}{\Stream$(G)$,$v$,$k$,$d$,$\Lambda$}	
		\If{$k=0$}
		\State{\Return F:=\{v\}}
		\EndIf
		\State create a graph $F$ with root vertex $v$
		\For{$(u,w)\gets$ next edge in the stream}
		\If{$(u,w)$ is a violating edge for $F$}
		\State{\Return \Bad}
		\ElsIf{$\radius(F+\{(u,w)\})\leq k$ and $\max\deg(F+\{(u,w)\})\leq d+1$}
		\State $V(F)\gets V(F)\cup\{u,w\}$
		\State $E(F)\gets E(F)\cup\{(u,w)\}$
		\State record the time-step $t_{(u,w)}$ of $(u,w)$
		\EndIf
		\EndFor
		\State $t_\ell\gets$ the time-step of the last edge added to $F$ 
		\If{$t_\ell > \Lambda$}
		\State \Return \Bad
		\Else
		\State \Return $F$
		\EndIf
		\EndProcedure
	\end{algorithmic}
\end{algorithm}

Note that if $v$ has extended $(d+1)$-bounded $k$-disc $\CD_{k,d}(v)\cong \Gamma$, and all edges are ordered in the canonical $\CD_{k,d}(v)$-edge ordering and appear in the first phase (i.e., the first $\Lambda$ edges) of the stream, then $\CD_{k,d}(v)$ will be fully collected and returned by the algorithm. 

On the other hand, if a vertex $u$ has $\CD_{k,d}(u)\ncong\Gamma$, then it still could happen that the observed graph contains a subgraph that is isomorphic to $\Gamma$, since some edges incident to high degree vertices (and also some edges incident to low degree vertices) appear earlier than they have been collected in the graph $F$. In general, if the algorithm does not return \Bad, we call the returned graph $F$ a \emph{lexicographical extended $(d+1)$-bounded $k$-disc (with respect to $\Gamma$)}, as the algorithm always collects edges in a lexicographical order (over the labels of corresponding vertices). Therefore, $\CD_{k,d}(v)$ is a special case of lexicographical disc.

\subsection{Approximating the Frequency of Extended $(d+1)$-Bounded $k$-Disc Type $\Gamma$}

Let $k,d$ be any integer. Let $\calJ_{k,d+1}$ denote the set of all $(d+1)$-bounded $k$-disc isomorphic types. Let $J=J(k,d)$ denote the number of types in $\calJ_{k,d+1}$. Note that $J\leq 2^{O((d+1)^{k+1})}$. Let $\Gamma\in \calJ_{k,d+1}$. Let $\gamma_\Gamma$ denote the probability that all edges in $\Gamma$ appears in the canonical order in the first phase of the stream, which again is determined by a threshold $\Lambda$ that is sampled from $\Bi(m,\tau)$. Note that $\gamma_\Gamma:=\frac{1}{t!}\cdot \tau^{t}$, where $t:=|E(\Gamma)|$. 



The details are described in Algorithm~\ref{alg:approx_disc_type}.
\begin{algorithm}[H]
	\caption{Approximating the frequency of extended $(d+1)$-bounded $k$-disc types}~\label{alg:approx_disc_type}
	\begin{algorithmic}[1]
		\Procedure{NumDisc}{\Stream$(G)$,$k,d$,$\delta,\rho$}	
		\State $\tau\gets O(\frac{\rho\delta}{J}\cdot (2d)^{-4k(d+1)^{2k}})$ 
		\State choose a number $\Lambda$ from $\Bi(m,\tau)$
		\State sample a set $A$ of $s:=\Theta((\frac{J}{\rho\delta})^{(d+1)^{2k}}\cdot (2d)^{4k(d+1)^{4k}})$ vertices uniformly at random
		\For{each $\Gamma\in \calJ_{k,d+1}$}
		\For{each $v\in A$}
		\If{\textsc{Stream\_CanoDisc}(\Stream($G$),$v$,$k$,$d$,$\Lambda$) returns $\Gamma$}
		
		\State $X_v=1$
		\Else
		\State{$X_v=0$}
		\EndIf
		\EndFor
		\State $\gamma_\Gamma\gets \frac{1}{t!}\cdot \tau^{t}$, where $t=|E(\Gamma)|$
		\State $C_\Gamma=\frac{\sum_{v\in A}X_v}{|A|}\cdot \frac{n}{\gamma_\Gamma}$
		\EndFor
		\State \Return $\{C_\Gamma: \Gamma\in \calJ_{k,d+1}\}$
		\EndProcedure
	\end{algorithmic}
\end{algorithm}


Our theorem is as follows.
\begin{theorem}~\label{thm:bounded_discs}
	Let $k,d,\delta,\rho>0$. There exists a random order streaming algorithm that with probability at least $1-\rho$, outputs the estimators $C_\Gamma$ such that 
	$$|C_\Gamma-\cc_\Gamma|\leq \delta n,$$
	for any $\Gamma\in \calJ_{k,d+1}$. The space complexity of the algorithm is $O((\frac{J}{\rho\delta})^{(d+1)^{2k}}\cdot (2d)^{4k(d+1)^{4k}})$.  
\end{theorem}

Before we present the proof of the above theorem, we use it to prove Theorem~\ref{thm:planar_mis}.
\begin{proof}[Proof of Theorem~\ref{thm:planar_mis}]
	In order to approximate the maximum independent set in the input minor-free graph with an additive error $\varepsilon n$ with probability $2/3$, we only need to set $d=O(1/\varepsilon)$ and $k=(1/\varepsilon)^{O(\log^2(1/\varepsilon))}$, and approximate the number of each $d$-bounded $k$-disc type $\Delta$ with an additive error $\delta n= \Theta(\varepsilon n)$. (See the discussion at the beginning of Section~\ref{sec:mis_planar}). This in turn can be done by approximating the number of each extended $(d+1)$-bounded $k$-disc type $\Gamma$ with the same additive error, for which we invoke Algorithm~\ref{alg:approx_disc_type} with the corresponding parameters. Note that all the parameters only depend on $\varepsilon$. Therefore, by Theorem~\ref{thm:bounded_discs} and that $J\leq 2^{O((d+1)^{k+1})}$, the space complexity of the algorithm is $2^{(1/\varepsilon)^{(1/\varepsilon)^{\log^{O(1)} (1/\varepsilon)}}}$. Finally, we note that the above algorithm easily yields a $(1+\varepsilon)$-approximation algorithm for the size of maximum independent set in any minor-free graph. This is true since the size of the maximum independent set in any such a graph is $\Omega(n)$.
\end{proof}

\subsection{Proof of Theorem~\ref{thm:bounded_discs}}
Now we present the proof of Theorem~\ref{thm:bounded_discs}. The space complexity can be analyzed similarly as in the proof of Theorem~\ref{thm:num_cc}. In the following, we focus on the correctness of the algorithm.

We first introduce the following definitions. For any lexicographical $(d+1)$-bounded $k$-disc $D$ collected from $v$ with respect to $\Gamma$ and any vertex $u$ with degree $d+1$ in $D$, we call an edge $(u,x)$ redundant for $u$ if it appears later than the last edge incident to $v$ we added to $D$. Note that since $D$ is lexicographical, it has been observed by Algorithm~\ref{alg:stream_k_disc} and that every redundant edge $(u,x)$ for a high degree vertex $u$ must satisfy that $\lab(x)>\lab(y)$, where $(u,y)$ is the last edge (i.e., the $(d+1)$-th edge) incident to $u$ that has been added to $D$.

We let $\D$ denote the subgraph that is constructed on the static graph $G$ in a way similar to Algorithm~\ref{alg:bounded_disc} without the restriction of collecting at most $d+1$ edges from each vertex, until we see the last edge that we added to $D$ obtained from Algorithm~\ref{alg:stream_k_disc} \textsc{Stream\_CanoDisc}. Let $H_D$ denote the number of vertices with degree $d+1$ in $D$, and let $\RE(H_D)$ denote all the redundant edges incident to vertices in $H_D$. Finally, we let $\BE(D)$ denote the set of edges in $E[\D]\setminus (\RE(H_D)\cup E[D])$ and we call edges in $\BE(D)$ \emph{Bad} edges with respect to $D$. One important property of $\BE(D)$ is that every edge in $\BE(D)$ must appear in the first phase (and also in the ``non-collectable'' position), that is, no edge in $\BE(D)$ was detected as a violating edge by \textsc{Stream\_CanoDisc}.

Note that if the canonical disc of $v$ is observed, i.e., $D=\CD_{k,d+1}(v)$, then $\BE(D)=\emptyset$. On the other hand, if $D$ is not the same as $\CD_{k,d+1}(v)$, then $\BE(D)$ must be non-empty.

\begin{definition}
	Let $\Gamma$ be any extended $(d+1)$-bounded $k$-disc type. We let $\calD_\Gamma(v)$ denote the set of all lexicographical discs rooted at $v$ with respect to $\Gamma$. We let $\calD_{\Gamma,b}(v)$ denote the set of all lexicographical discs $D$ rooted at $v$ with respect to $\Gamma$ and with $|\BE(D)|=b$. Let $y_{\Gamma,b}=|\calD_{\Gamma,b}(v)|$.
\end{definition}

We have the following lemma.
\begin{lemma}~\label{lemma:num_disc}
	Let $v$ be a vertex with degree at most $d$ and let $\Gamma$ be any extended $(d+1)$-bounded $k$-disc type. Then it holds that for any $b\geq 0$, 
	$$y_{\Gamma,b}\leq \kappa(d,b,k):=(d+b+1)^{3k(d+1)^{2k}}.$$
\end{lemma}
\begin{proof}
	Let us consider any lexicographic disc $D\in \calD_{\Gamma,b}(v)$. Note that besides vertices in $H_D$, all the other vertices in $D$ must have degree at most $b+d+1$. On the other hand, for each high degree vertex $u$ in $H_D$, only the first at most $b+d+1$ edges incident to $u$ (according to the lexicographical order over the other endpoints) can appear in $D$, as otherwise there will be more than $b$ redundant edges for $u$, which is a contradiction. 
	
	If we let $K_v$ denote subgraph that only contains vertices and edges that could be part of some lexicographic disc $D$ with $|\BE(D)|=b$, which can be constructed by Algorithm~\ref{alg:bounded_disc} with parameters $G,v,k,d+b$. Note that the number of vertices in $K_v$ can be at most $1+(d+b+1)+\cdots+(d+b+1)^k\leq (d+b+1)^{k+1}$ vertices. The total number of edges in $K_v$ is thus at most $(d+b+1)^{2(k+1)}$. Finally, the total number of subgraph $S$ of $K_v$ that is isomorphic to $\Gamma$ is at most $\binom{(d+b+1)^{2(k+1)}}{|E(\Gamma)|}\leq (d+b+1)^{3k(d+1)^{2k}}$ as $|E(\Gamma)|\leq (d+1)^{k+2}$. This further implies that $y_{\Gamma,b}\leq (d+b+1)^{3k(d+1)^{2k}}$.
\end{proof}

Let $v\in A$. We again write $A\sim \calU_V$ to indicate that $A$ is a set of $s$ vertices sampled uniformly at random from $V$. We write $\sigma \sim\calU_{\calP[E]}$ to indicate that an edge ordering $\sigma$ is sampled from the uniform distribution over the set $\calP[E]$ of all permutations of edges. For any extended $(d+1)$-bounded $k$-disc type $\Gamma$, we let $K_v(b,\Gamma)$ denote the subgraph that only contains vertices and edges that could be some lexicographic extended disc $D$ with respect to $\Gamma$ and with $|\BE(D)|=b$. Note that from the proof of Lemma~\ref{lemma:num_disc}, we know that for $|V(K_v(b,\Gamma))|\leq (d+b+1)^{k+1}$.

Now for any $(d+1)$-bounded $k$-disc type $\Gamma$, we let $A_\Gamma\subseteq A$ denote the subset of vertices $v$ in $A$ such that $\CD_{k,d+1}(v)\cong\Gamma$. We have the following lemma. 

\begin{lemma}~\label{lemma:A_for_disc}
	With probability (over the randomness of $A\sim \calU_V$) at least $1-\frac{\rho}{5}$, it holds that for any $\Gamma$, $|\frac{|A_\Gamma|}{|A|}-\frac{\cc_\Gamma}{n}|\leq \frac{\delta}{4}$; and for any two vertices $u,v\in A$ with $\CD_{k,d+1}(u)\cong \Gamma$ and $\CD_{k,d+1}(v)\cong \Gamma$, $K_u(b_1,\Gamma)$ has no intersection with $K_v(b_2,\Gamma)$, for any $b_1,b_2\leq b_0:=5k(d+1)^{2k}\log_{1/\gamma_\Gamma}(n)$.
\end{lemma}
\begin{proof}
	The first part is a straightforward application of Chernoff bound. For the second part, it suffices to note that $|\bigcup_{b\leq b_0}V(K_v(b,\Gamma))| /n \leq o_n(1)$ by our setting of $b_0$.
\end{proof}

In the following, we will condition on the above two events about the sampled set $A$ in Lemma~\ref{lemma:A_for_disc}, which occur with probability at least $1-\frac{\rho}{5}$. 

Now let us consider a fixed disc type $\Gamma$ and our estimator $C_\Gamma$ for the number of extended $(d+1)$-bounded $k$-disc type $\Gamma$. For any vertex $v\in A$, we distinguish the following two cases.

\textbf{(\rom{1}) $\CD_{k,d+1}(v,G)\cong \Gamma$.} In this case $X_v=1$ will happen if either 1) the canonical disc $\CD_{k,d+1}(v,G)$ is observed in the first part of the stream, or 2) some other lexicographical disc $D$ is observed such that $D \cong \Gamma$ while $D$ is different from $\CD_{k,d+1}(v,G)$. Note that the first event happens with probability $\gamma_{\Gamma}$. 

For the second event, we note that since $D$ is observed and $D\cong \Gamma$, this means all the edges in $D$ appeared in some lexicographical order and in the first phase of the stream, which happens with probability $\gamma_{\Gamma}$. Now for each such $D$ with the corresponding extended $\D$ (see the above definition) and $|\BE(D)|=b$, it holds that $b\geq 1$, and all the edges in $\BE(D)$ appeared in the ``non-collectable'' order in the first phase of the stream, which happens with probability at most $\tau^b$. If we let $\calE_{D}$ denote the event that the lexicographical disc $D$ such that $D\cong \Gamma$ has been observed and let $\eta=\Theta(\rho\delta/J)$, then by Lemma~\ref{lemma:num_disc}, the probability that the second event occurs is
\begin{eqnarray}
\Pr_{\sigma \sim\calU_{\calP[E]}}[\bigcup_{\substack{D\sim \Gamma\\
		\textrm{$D$ is not canonical}}} \calE_D]
&\leq&\bigcup_{\substack{D\sim \Gamma \\
		\textrm{$D$ is not canonical}}}\Pr_{\sigma \sim\calU_{\calP[E]}}[\calE_D] \nonumber \\
&\leq & \sum_{b\geq 1} \tau^b (d+b+1)^{3k(d+1)^{2k}} \cdot \gamma_\Gamma \nonumber\\ &\leq& \eta\cdot \gamma_\Gamma, \label{eqn:bad_disc_event}
\end{eqnarray}
where the last inequality follows from similar calculation as we have done in the proof of Theorem~\ref{thm:num_cc} and our current choice of parameters.

Therefore, in this case,
$$\gamma_\Gamma\leq \Pr_{\sigma \sim\calU_{\calP[E]}}[X_v=1]\leq (1+\eta)\gamma_\Gamma.$$

%
%
%
%
%
%
%
%
%
%
%
%
%
%
%

\textbf{(\rom{2}) $\CD_{k,d+1}(v,G)\ncong \Gamma$.} Then $X_v=1$ will happen if some lexicographical disc $D$  has been observed such that $D\cong \Gamma$. By our assumption, any such $D$ must satisfy that $|\BE(D)|\geq 1$. Thus, the probability that such an event occurs with the same probability as given by formula~(\ref{eqn:bad_disc_event}). Therefore, in this case 
$$\Pr_{\sigma \sim\calU_{\calP[E]}}[X_v=1]\leq \eta \gamma_\Gamma.$$

Now let $Y_1:=\sum_{v\in A_\Gamma} X_v$. It holds that $\E_{\sigma\sim\calU_{\calP[E]}}[Y_1]=\sum_{v\in A_\Gamma} \E_{\sigma\sim\calU_{\calP[E]}}[X_v]$ and thus that $|A_\Gamma|\gamma_\Gamma\leq \E_{\sigma\sim\calU_{\calP[E]}}[Y_1]\leq (1+\eta)|A_\Gamma|\gamma_\Gamma$. Furthermore, recall that we have conditioned on the event that for any two vertices $u,v\in A_\Gamma$, $K_u(b_1,\Gamma)$ has no intersection with $K_v(b_2,\Gamma)$, for any $b_1,b_2\leq b_0$. Now for any vertex $u\in A_\Gamma$, we let $\calE_{D_u}$ denote the event that the lexicographical disc $D_u$ such that $D_u\sim \Gamma$ has been observed. For any two vertices $u,v\in A_\Gamma$, we let $\calE_{D_u,D_v}$ denote the event that both lexicographical discs $D_u$ and $D_v$ such that $D_u\sim \Gamma$ and $D_v\sim \Gamma$ are observed. 
\begin{eqnarray*}
	&&\Pr_{\sigma \sim\calU_{\calP[E]}}[X_uX_v=1] 
	\\&=&
	\Pr_{\sigma\sim\calU_{\calP[E]}}[\bigcup_{D_u\in \calD_{\Gamma}(u), D_v\in \calD_{\Gamma}(v)}\calE_{D_u,D_v}]\\
	&\leq&
	 \sum_{D_u\in \calD_{\Gamma}(u), D_v\in \calD_{\Gamma}(v)}\Pr_{\sigma\sim\calU_{\calP[E]}}[\calE_{D_u,D_v}]\\
	&=& 
	\sum_{\substack{D_u\in \calD_{\Gamma}(u), D_v\in \calD_{\Gamma}(v)\\ \textrm{$D_u,D_v$ independent}}}\Pr_{\sigma\sim\calU_{\calP[E]}}[\calE_{D_u,D_v}] 	 
	+ \sum_{\substack{D_u\in \calD_{\Gamma}(u), D_v\in \calD_{\Gamma}(v)\\ \textrm{$D_u,D_v$ dependent}}}\Pr_{\sigma\sim\calU_{\calP[E]}}[\calE_{D_u,D_v}]\\
	&\leq &
	\sum_{\substack{D_u\in \calD_{\Gamma}(u), D_v\in \calD_{\Gamma}(v)\\ \textrm{$D_u,D_v$ independent}}}\Pr_{\sigma\sim\calU_{\calP[E]}}[\calE_{D_u}]\Pr_{\sigma\sim\calU_{\calP[E]}}[\calE_{D_v}]
	+ \sum_{\substack{D_u\in \calD_{\Gamma,b}(u), D_v\in \calD_{\gamma,b'}(v)\\ b,b'\geq b_0}} \Pr_{\sigma\sim\calU_{\calP[E]}}[\calE_{D_v}]\\
	&\leq& \sum_{\substack{D_u\in \calD_{\Gamma}(u), D_v\in \calD_{\Gamma}(v)\\ \textrm{$D_u,D_v$ independent}}}\Pr_{\sigma\sim\calU_{\calP[E]}}[\calE_{D_u}]\Pr_{\sigma\sim\calU_{\calP[E]}}[\calE_{D_v}] 
	+ \sum_{b\geq b_0} \kappa(d,b,k)\sum_{b'\geq b_0} \kappa(d,b,k)\cdot \gamma_\Gamma^{b'}\\
	&\leq& (\sum_{D_u\in \calD_{\Gamma}(u)}\Pr_{\sigma\sim\calU_{\calP[E]}}[\calE_{D_u}])^2 + n^{4k(d+1)^{2k}}\cdot \gamma_\Gamma^{b_0}\\
	&\leq& ((1+\eta)\cdot \gamma_\Gamma)^2 + o_n(1)
\end{eqnarray*}
where the last inequality follows from the fact that $\calD_{\Gamma}(u) = \{D:D\cong\Gamma, \textrm{$D$ is not canonical}\}\cup \CD_{k,d+1}(u)$, the inequality~(\ref{eqn:bad_disc_event}) and that $b_0=5k(d+1)^{2k}\log_{1/\gamma_\Gamma}(n)$.

Therefore, 
\begin{eqnarray*}
	\Var_{\sigma\sim\calU_{\calP[E]}}[Y_1] &=&	\E_{\sigma\sim\calU_{\calP[E]}}\left[Y_1^2\right] - \E_{\sigma\sim\calU_{\calP[E]}}\left[Y_1\right]^2\\
	&\le&\E_{\sigma\sim\calU_{\calP[E]}}\left[(\sum_{v\in A_\Gamma}  X_v)^2\right] - |A_\Gamma|^2 \gamma_\Gamma^2\\
	&=&\E_{\sigma\sim\calU_{\calP[E]}}\left[\sum_{u,v\in A_\Gamma}X_uX_v\right] - |A_\Gamma|^2 \gamma_\Gamma^2 \\
	&=&\E_{\sigma\sim\calU_{\calP[E]}}\left[\sum_{u\in A_\Gamma}X_u^2 + \sum_{v\in A_\Gamma}\sum_{u\in A_\Gamma\setminus\{v\}}X_uX_v\right]
	 - |A_\Gamma|^2 \gamma_\Gamma^2\\
	&\leq& \sum_{u\in A_\Gamma} \E_{\sigma\sim\calU_{\calP[E]}}[X_u^2]+ |A_\Gamma|^2  ((1+\eta)\cdot \gamma_\Gamma)^2 +o_n(1)
	 - |A_\Gamma|^2 \gamma_\Gamma^2\\
	&\leq& |A_\Gamma|\cdot \gamma_\Gamma + |A_\Gamma|^2\cdot \gamma_\Gamma^2 \cdot  ((1+\eta)^2-1) + o_n(1) \\
	&\leq& |A_\Gamma|\cdot \gamma_\Gamma + 3 \eta \cdot |A_\Gamma|^2 \gamma_\Gamma^2,
\end{eqnarray*}
where the last inequality follows from our choice of parameters.

Thus, by Chebyshev inequality, 
\begin{eqnarray*}
	\Pr_{\sigma\sim\calU_{\calP[E]}}[|Y_1-\E[Y_1]|\geq \frac{\delta \gamma_\Gamma|A|}{4}] 
	\leq  \frac{16\cdot (\gamma_\Gamma|A_\Gamma| + 3\eta |A_\Gamma|^2 \gamma_\Gamma^2)}{\delta^2 |A|^2\gamma_\Gamma^2} 
	&\leq& \frac{\rho}{10 J}, 
\end{eqnarray*}
where in the last inequality, we used the fact that $|A_\Gamma|\leq |A|$ and that $|A|=\Theta((\frac{J}{\rho\delta})^{(d+1)^{2k}}\cdot (2d)^{4k(d+1)^{4k}})=\Omega(\frac{J}{\rho\delta^2\gamma_\Gamma})$.
This further implies that with probability at least $1-\frac{\rho}{10J}$, 
$$|A_\Gamma|\gamma_\Gamma - \frac{\delta \gamma_\Gamma|A|}{4}\leq Y_1\leq(1+\varepsilon)|A_\Gamma|\gamma_\Gamma +\frac{\delta \gamma_\Gamma|A|}{4}.$$

Now let $Y_2:=\sum_{v\in A\setminus A_\Gamma} X_v$. Then $\E_{\sigma\sim\calU_{\calP[E]}}[Y_2]=\E_{\sigma\sim\calU_{\calP[E]}}[\sum_{v\in A\setminus A_\Gamma}X_v]\leq |A|\cdot \eta\cdot \gamma_\Gamma$, where the last inequality holds since $\E_{\sigma\sim\calU_{\calP[E]}}[X_v]\leq \eta\gamma_\Gamma$ for any $v\in A\setminus A_\Gamma$.

By Markov's inequality, we have that 
\begin{eqnarray*}
	\Pr_{\sigma\sim\calU_{\calP[E]}}[Y_2\geq \frac{\delta \gamma_\Gamma|A|}{4}] \leq \frac{4\cdot \E[Y_2]}{\delta |A|\gamma_\Gamma} \leq \frac{\rho}{10J},
\end{eqnarray*}
where in the last inequality we used the fact that $\eta=O(\rho\delta/J)$.

Therefore, with probability at least $1-\frac{\rho}{5J}$, 
\begin{eqnarray*}
	|A_\Gamma|\gamma_\Gamma - \frac{\delta \gamma_\Gamma|A|}{4}
	\leq 
	Y_1+Y_2\leq 
	(1+\varepsilon)|A_\Gamma|\gamma_\Gamma +\frac{\delta \gamma_\Gamma|A|}{2}
\end{eqnarray*}

Now since $C_\Gamma=\frac{\sum_{v\in A}X_v}{|A|}\frac{n}{\gamma_\Gamma}=\frac{Y_1+Y_2}{|A|}\frac{n}{\gamma_\Gamma}$, we have that with probability at least $1-\frac{\rho}{5J}$, 
\begin{eqnarray*}
C_\Gamma&\geq& \frac{|A_\Gamma|\gamma_\Gamma-\frac{\delta|A|\gamma_\Gamma}{4}}{|A|}\frac{n}{\gamma_\Gamma}
\geq \left(\frac{\cc_\Gamma\cdot \gamma_\Gamma}{n}-\frac{\delta\gamma_\Gamma}{4}-\frac{\delta\gamma_\Gamma}{4}\right)\frac{n}{\gamma_\Gamma}
\geq
 \cc_\Gamma-\frac{1}{2}\delta n,
\end{eqnarray*}
 and 
\begin{eqnarray*}
C_\Gamma\leq \frac{(1+\eta)|A_\Gamma|\gamma_\Gamma+\frac{\delta|A|\gamma_\Gamma}{2}}{|A|}\frac{n}{\gamma_\Gamma}
&\leq&
 \left(\frac{(1+\eta)\cdot \cc_\Gamma\cdot \gamma_\Gamma}{n}+\frac{\delta\gamma_\Gamma}{4}+\frac{\delta\gamma_\Gamma}{2}\right)\frac{n}{\gamma_\Gamma}\\
&\leq& \cc_\Gamma+\delta n,
\end{eqnarray*}
where in the above inequalities, we used the condition that $|\frac{|A_\Gamma|}{|A|}-\frac{\cc_\Gamma}{n}|\leq \frac{\delta}{4}$.

Therefore, with probability at least $1-\frac{\rho}{5J}\cdot J=1-\frac{\rho}{5}$, for all $\Gamma\in \calJ_{d,k+1}$, $$\cc_\Gamma-\frac{1}{2}\delta n\leq C_\Gamma\leq \cc_\Gamma + \delta n.$$

The statement of the theorem then follows by taking the union bound of the error probabilities over the randomness of $A\sim \calU_V$ and $\sigma\sim\calU_{\calP[E]}$.

\section{Conclusions and Open Problems}

We have introduced a new technique to transform constant time algorithms into random order
streaming algorithms. It would be interesting to see other examples when our technique can be
used. Interesting candidates are for example random walk based algorithms. For example, it
could be interesting to estimate the return probabilities of random walks in random order streams.
It would also be interesting to show that any constant time property tester for graphs with
bounded average degree (or graph classes with bounded average degree) can be simulated
in random order streams. In order to prove such a statement it is first needed to define a canonical
property tester for graphs that can contain vertices of arbitrary degree.

Besides getting new upper bound, it will also be interesting to obtain lower bounds for random 
order streams. It seems to be plausible to conjecture that approximating the number of connected
components requires space exponential in $1/\epsilon$. It would be nice to have lower bounds that
confirm this conjecture.

\bibliographystyle{alpha}
\bibliography{literature}

\end{document}